\documentclass{article}
\usepackage[accepted]{icml2019}

\usepackage{microtype}
\usepackage{graphicx}
\usepackage{subfigure}
\usepackage{booktabs} 
\usepackage{amsmath, amsthm} 
\usepackage{amssymb}
\usepackage{url}

\usepackage{hyperref} 
\usepackage{amsfonts}

\usepackage[dvipsnames]{xcolor}

\definecolor{light-gray}{gray}{0.95}
\newcommand{\code}[1]{\colorbox{light-gray}{\texttt{#1}}}
\newcommand{\ynot}{\overline{y}}
\newcommand{\alphanot}{\overline{\alpha}}

\newcommand{\mlark}{Masked LARk}
\DeclareMathOperator*{\argmin}{arg\,min}

\newtheorem{theorem}{Theorem}[section]
\newtheorem{corollary}{Corollary}[theorem]
\newtheorem{lemma}[theorem]{Lemma}
\newtheorem{definition}[theorem]{Definition}

\begin{document}

\twocolumn[
\title{Masked LARk: \\
       Masked Learning, Aggregation and Reporting worKflow}
\author{Joseph J. Pfeiffer III, Microsoft, joelpf@microsoft.com \\ 
        Denis Charles, Microsoft, cdx@microsoft.com \\
        Davis Gilton, Microsoft, davisgilton@microsoft.com \\
        Young Hun Jung, Microsoft, youjung@microsoft.com \\
        Mehul Parsana, Microsoft, mparsana@microsoft.com \\
        Erik Anderson, Microsoft, erikan@microsoft.com \\}
\date{\vspace{-0.2in}}
\maketitle

\begin{abstract}
Today, many web advertising data flows involve passive cross-site tracking of users. Enabling such a mechanism through the usage of third party tracking cookies (3PC) exposes sensitive user data to a large number of parties, with little oversight on how that data can be used.  Thus, most browsers are moving towards removal of 3PC in subsequent browser iterations.  In order to substantially improve end-user privacy while allowing sites to continue to sustain their business through ad funding, new privacy-preserving primitives need to be introduced.

In this paper, we discuss a new proposal, called \textit{\mlark}, for aggregation of user engagement measurement and model training that prevents cross-site tracking, while remaining (a) flexible, for engineering development and maintenance, (b) secure, in the sense that cross-site tracking and tracing are blocked and (c) open for continued model development and training, allowing advertisers to serve relevant ads to interested users.  We introduce a secure \textit{multi-party compute} (MPC) protocol that utilizes "helper" parties to train models, so that once data leaves the browser, no downstream system can individually construct a complete picture of the user activity.  For training, our key innovation is through the usage of masking, or the obfuscation of the true labels, while still allowing a gradient to be accurately computed in aggregate over a batch of data.  Our protocol only utilizes light cryptography, at such a level that an interested yet inexperienced reader can understand the core algorithm.  We develop helper endpoints that implement this system, and give example usage of training in PyTorch.

\end{abstract}

\vspace{0.4in}
]
\clearpage

\section{Introduction}

Today, many web advertising data flows involve passive cross-site tracking of users through the use of third party cookies (3PC).  Such mechanisms enable many valid use cases, such as reporting advertiser {\em conversion} metrics, which can be used to identify which campaigns are more or less likely to lead to desired behavior by a user. Examples of conversions include landing on a page of interest, putting an item into a shopping cart, or making a purchase.  More advanced techniques deploy machine learning models to identify the likelihood of a conversion, and adjust an advertiser's corresponding bid in real-time.

However, a strong drawback to 3PC is that they expose sensitive user data to a large number of parties, with little oversight on how that data can be used.  Thus, most browsers are moving towards removal of 3PC in subsequent browser iterations \cite{apple2020threepc,firefox2019threepc,chrome2021threepc}.  In order to substantially improve end-user privacy while allowing sites to continue to sustain their business through ad funding, new privacy-preserving primitives need to be introduced.  Several proposals \cite{google2020aggregatereporting,corrigangibbs2017prio} have been introduced, primarily relying on the notion of a {\em trusted} third party charged with aggregating and reporting metrics back to the ad server.  This brings up questions of how this third party can be trusted not to collude with various ad servers. Some ideas exist, ranging from code audits to secure enclaves \cite{anciaux2019enclaves}, but most have either large overhead or expense.

In contrast, a recent proposal from Google goes in a different direction: utilizing secure multi-party compute (MPC) to build a trusted third party comprised of multiple semi-trusted helpers \cite{google2020aggregatereporting}.  Each helper receives a piece of information necessary to complete a larger picture, but without the complementary pieces held by other helpers the data is effectively random. The primary assumption in this world is the lack of {\em collusion} -- so long as the helpers do not work together or with the advertising server (breaking the protocol), no single party can recover the original sample of information.  Despite this restriction, the helpers can follow assigned protocols to create useful {\em aggregate} values.  

An illustration of the system is shown in Figure \ref{fig:MPC}.  The data begins with a (large) collection of user data being held solely on the user's device, the browser.  The browser creates {\em secret shares} of the data intended for multiple helpers (e.g., $H0$ and $H1$).  These shares are constructed in such a way that an individual helper receives only a piece of the final value, which alone is effectively random.  The shares, or reports, are then encrypted with public keys provided by the helpers (keys are unique to the helper) and passed on to the advertiser server.  The advertising server effectively acts as a database -- there's very little available information that it can read, and it must rely on the helpers to construct meaningful aggregate information about the users.  How the secret shares are constructed, and what is hidden, is central to the overarching design of the system.  A popular approach is to use protocols which involve hiding both context about the conversion and the conversion value (e.g., using {\em distributed point functions} \cite{google2020aggregatereporting}).  

\begin{figure}[t]
\begin{center}
\centerline{\includegraphics[width=\columnwidth]{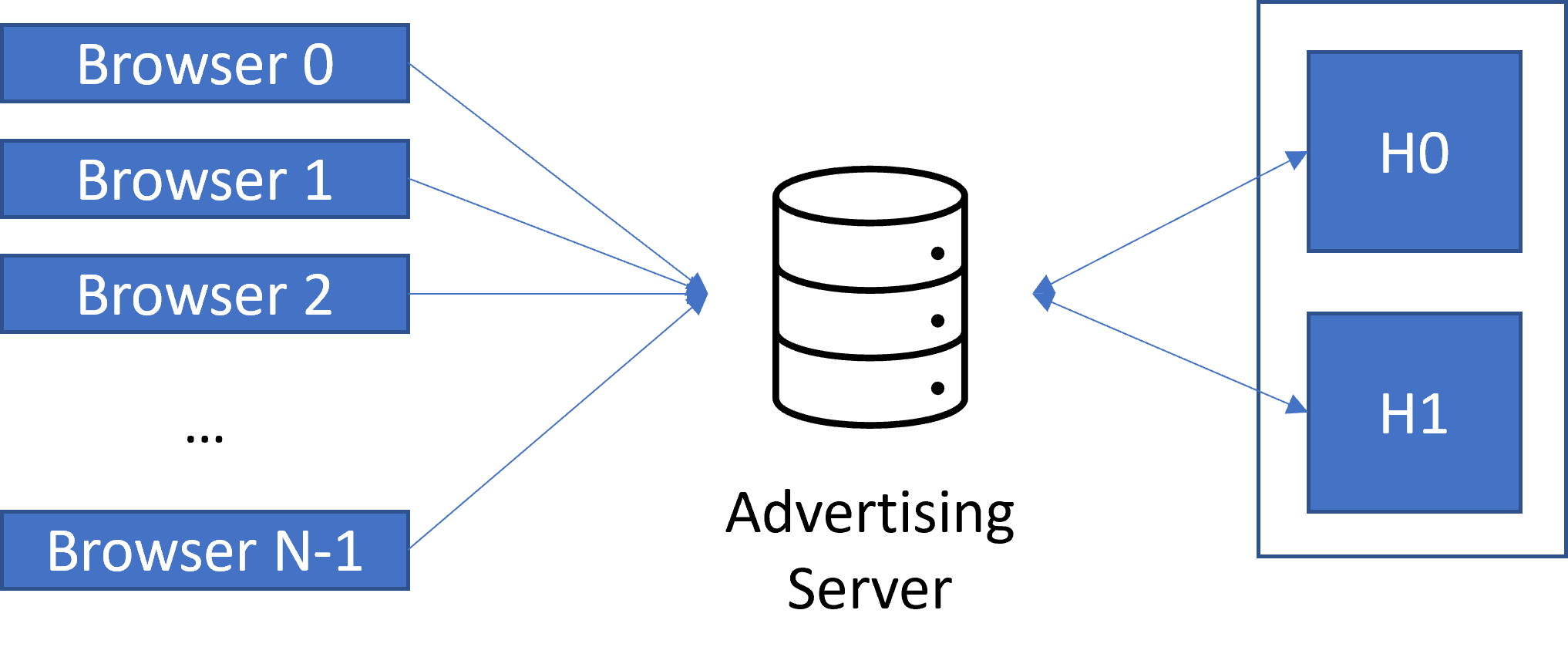}}
\caption{System overview: browsers track feature vectors and send {\em encrypted} reports to an advertising server.  The advertising server sends the reports to the respective helpers, who decrypt and perform local aggregation.  Noisy local aggregate reports are returned to the advertising server, who combines them for a complete (noisy) value.}
\label{fig:MPC}
\end{center}
\end{figure}

Here, we propose Masked LARk, a Masked Learning, Aggregation and Reporting worKflow.

\begin{itemize}
    \item We extend the existing proposal \cite{google2020aggregatereporting} to a variety of tasks outside of simple aggregation reporting.  We envision Masked LARk as a privacy preserving Map-Reduce platform, where browsers map user information into secret shared representations, and secure MPC helpers perform a private reduce operation.
    \item We discuss how differentially private \cite{dwork2014algorithmic} aggregated statistics can be computed between browser and helpers and sent to the advertising server.  
    \item We outline the usage of Masking, and show how it generalizes to a wide variety of tasks.  Notably, this allows us to compute gradients for differentiable models, and train conversion models using secure MPC.
    \item In Section \ref{subsubsec:localdp}, we show if we utilize local differential privacy on our feature vectors, the data sent to the helpers is differentially private, and by extension the learning algorithm is also differentially private.
    \item We give several experiments on benchmark datasets, and plan to release a Python package with example usage of Masked LARk for model training.
\end{itemize}



The remainder of the paper is structured as follows: we start in Section \ref{sec:relatedwork} by discussing related work to this domain, including Google's proposal, an overview of MPC, and Federated Learning.  Next, we give an overview of the Masked LARk proposal in Section \ref{sec:masked_lark}, discussing some types of functions that Masked LARk can compute along with their privacy requirements.  In Section \ref{sec:gradient_computation}, we give an in-depth discussion of \emph{masking} and its usefulness for computation over bi-linear functions. This includes model training algorithms that require gradient computation.  In Section \ref{sec:experiments}, we show that we can indeed train models using MPC helpers, with ablation studies on the impact of privacy settings and local vs. global differential privacy.  Section \ref{sec:discussion} concludes the paper.  In addition, we currently have an externally-available software package for Masked LARk, available at\\ \url{https://github.com/microsoft/maskedlark}.

\section{Related Work}\label{sec:relatedwork}

Masked LARk sits within a community of proposals related to improving privacy for web advertising.  Notable related projects include FLoC \cite{google2021floc}, FLEDGE \cite{google2021fledge}, PARAKEET \cite{microsft2021parakeet} MACAW \cite{microsft2021macaw} or the Conversion Measurement proposal \cite{google2020aggregatereporting}, amongst others.  Of these, each set of proposals has a different functionality: FLoC was proposed as a way to utilize federated learning to create user clusters for privacy, which can then be utilized by advertisers at ad serving time. FLEDGE / PARAKEET / MACAW are proposals for running privacy-preserving auctions, where user signals are masked from advertising servers entirely while an auction mechanism determines which ads to show on a page.  

In contrast, both Masked LARk and the Google's proposal \cite{google2020aggregatereporting} work on data \emph{after} the impression, click or view has occurred.  Masked LARk largely builds onto the aggregated conversion measurement first brought forward by Google \cite{google2020aggregatereporting,corrigangibbs2017prio}.  The overarching structure between the two is the same; the browsers create encrypted reports intended for an MPC service, utilizing the advertising service itself as an intermediate data storage.  Masked LARk differs from this proposal in two notable ways.  First, Masked LARk has browsers insert random noise into the data samples sent to the advertising server, obfuscating real versus fake features when consumed by the helpers.  This is used to allow simpler downstream crypto protocols to mask only the conversion values themselves, even preventing the helper services from making connections from a feature vector to a label.  In contrast, the original proposal \cite{google2020aggregatereporting} relies on a hierarchical representation of bits requiring interactions between the helpers known as distributed point functions (DPF) \cite{boneh2020dpf}, which explicitly hides the features from the helpers as well.  Second, Masked LARk envisions a more general platform for computing many types of functions, rather than solely aggregating measurement.  A good example of such functions is the gradient computation function for training machine learning models.  Additional functions that can be computed through the Masked LARk are the DPF representations or other available protocols for secure multi-party aggregation \cite{msr2021bucketization}.  What functions to implement is a community decision, discussed among many relevant parties. 

The multi-party computation involves a rich history within cryptography, with many primitives available for use.  An accessible introduction to many of these primitives can be found in the book \cite{evans2018mpc}.  The protocols discussed within this work center around either additive or multiplicative sharing (masking), and their novelty lies in their application within the overall workflow and its usage for model training.  More complex model inference \cite{rathee2020cryptflow2} and model training \cite{wagh2019securenn} exist, with much stronger requirements on feature and model privacy.  Conversely, Masked LARk directly obfuscates the \emph{label}, relying on random noise to further obfuscate the identifying information.  This feature results in protocols with no overhead between the helpers, and much simpler and more accessible implementation.

An additional related field to our domain is \emph{federated learning} \cite{kairouz2019fl}, which is extremely decentralized training of machine learning models.  Many aspects of federated learning align well with our proposed framework; for example, when gradient updates are computed per user, these can be additive secret shared for later consumption by helpers for summation. Hence, the involved cryptography becomes exceptionally scalable and simple.  However, a primary drawback is complexity on the \emph{user's} device.  We have some unknown number of advertising services that would utilize the platform. These services would push the model updates to individual devices, each of which uses its own computing resources to perform the corresponding updates. Such a system would necessitate moving large amounts of data to and from users' devices at irregular intervals, and would increase the computational burden required to store and perform training updates. Hence, this is an unattractive option.

\section{Masked LARk}
\label{sec:masked_lark}

The Masked Learning, Aggregation and Reporting worKflow (Masked LARk) proposal aims to solve multiple items:

\begin{itemize}
\setlength\itemsep{-.5em}
    \item Various private function computation using secure multi-party computation (MPC)
    \item Flexibility in representation of additional functions as new protocols and needs arise
    \item Simplicity in execution, focusing on decoupling users' private information from the advertiser websites they visit and/or convert on.
\end{itemize}

There are three groups of parties central to this platform


\begin{itemize}
    \item {\em Browsers} represent the individual user browsing and interacting with various pieces of the web.
    \item {\em Advertising (Ad) Servers} run the auction and provide ad impressions to the browser to display to the user.
    \item {\em Helpers} form an MPC platform for privacy-preserving function evaluation and aggregation.
\end{itemize}

The defined workflow has similarities regardless of the protocol and function to be executed.  In effect, the workflow mirrors the heavily utilized MapReduce systems \cite{dean2008mapreduce} that are ubiquitous to today's computation.  Masked LARk envisions two analogous pieces:

\begin{itemize}
    \item {\em Map}: browsers take the private data and encode it into a) a format for later functions to consume and b) encrypted packages intended for MPC
    \item {\em Reduce}: a collection of records is sent to the MPC platform with a function to execute, along with side information and privacy constraints
\end{itemize}

The advertising server, acting as a database, remains unable to read even minimal data throughout, and nothing specific to a user.

\subsection{Browser Actions}
Similar to the related proposals, browsers are responsible for a number of tasks in the workflow, handling across-site attribution, secret sharing the records and sending to the helpers.  Within \mlark, specifically, there are two notable additional responsibilities:

\begin{itemize}
    \item Upon the expiration of a possible attribution, if no attribution has occurred, a default value (e.g., 0) must be sent.  This lets models have negative samples with which to train.  The sample does not need to be sent right after the expiration occurs, and should delay random additional intervals to prevent timing attacks.
    \item Browsers are charged with sending additional fake records to obfuscate true values.
\end{itemize}

Both of these are important for privacy preservation and reasonable model training.  In certain circumstances, browser-injected-noise can guarantee differential privacy, which we discuss in Section \ref{subsubsec:localdp}.

\subsection{Helper's Privacy Preserving Techniques}
In addition to purely doing computation on the secret shares,  helpers have some additional features that protect user privacy: {\em $k$-anonymity} and {\em differential privacy}.

\begin{definition}[$k$-Anonymity \cite{samarati2001protecting}]
A helper releases the aggregates only when there are more than $k$ records. 
\end{definition}

The $k$-anonymity ensures that aggregation is done only when there is sufficiently large data set. This procedure will prevent the case where the aggregated output conveys too much information about individual data points.  

We note that $k$ is measured on a per-record basis -- in theory, all records could be provided from the same user, but the helpers have no way to uniquely identify users (by design).  That is why we pair it with another technique, {\em Differential Privacy} \cite{dwork2014algorithmic}. In a high level, either the browser or the helper adds noise to the true value before sharing it with others so that the data receiving entity is hard to infer about individual records. Readers who are interested in the formal definition of differential privacy can refer to the above book. 

Depending on whether the noise is injected to the input or to the output, there are {\em local} differential privacy and {\em global} differential privacy:

\begin{itemize}
    \item {\em Local Differential Privacy}: browsers are charged with either modifying or inserting fake records, to break apart information that a helper could glean.  In some scenarios (see Section \ref{subsubsec:localdp}), we can enforce local differential privacy in the browser.
    \item {\em Global Differential Privacy}: a helper is charged with inserting noise into the operation for global differential privacy.
\end{itemize}

Note that $k$-anonymity requires $k$ as a parameter and differential privacy has its privacy budget $\epsilon$. These privacy settings are publicly declared per advertising server, and are used by the browsers and the helpers to enforce the privacy mechanism.

\subsection{Example: Summation And Counting}\label{subsec:summations}
A simple example that fits within Masked LARk is summation (or counting, a simpler variant).  Here, the browser is given a ring over some modulo on the integers, i.e., $\mathcal R = \mathbb Z / m \mathbb Z$, and a value to return $y$.  Publicly declared privacy settings include

\begin{itemize}
    \item $k$, the minimum number of records required to perform private computation
    \item $\epsilon$, the privacy budget
    \item $\Delta_y$, the bounds of the value being summed, or the maximum change a single record can have on the final summation. If $y$ is a value from a single user, then $\Delta_y = \Delta_y^{max} - \Delta_y^{min}$, where $\Delta_y^{min} \leq y \leq \Delta_y^{max}$ for all $y$
\end{itemize}

\begin{algorithm}[t]
    \caption{Additive Secret Sharing}\label{alg:secretshare}
    \begin{algorithmic}
        \STATE \textbf{Inputs}: Field $\mathcal R$, value $v$
		\STATE Draw $\psi^0 \sim Unif(\mathcal R)$
		\STATE Compute $\psi^1 = v - \psi^0$
		\STATE Encrypt-And-Send $\psi^0$ to $H0$
		\STATE Encrypt-And-Send $\psi^1$ to $H1$
    \end{algorithmic}
\end{algorithm}

\begin{algorithm}[t]
    \caption{Helper Partial Recovery of Secret Shares}\label{alg:summation_recovery}
    \begin{algorithmic}
        \STATE \textbf{Inputs}: Privacy Parameters $k$ and $\epsilon$, sensitivity of function $\Delta_y$, helper records $R^h$
        \IF {LEN($R^h$) $<$ $k$}
        \STATE Return 0
        \ENDIF
        \STATE Set $\Psi^h = 0$
        \FOR {$r^h$ in $R^h$}
           \STATE $\psi^h = $ Decrypt($r^h$)
           \STATE Update $\Psi^h = \Psi^h + \psi^h$
        \ENDFOR
       \STATE Perturb the output $\tilde\Psi^h = \Psi^h + Laplace(0, \Delta_y / \epsilon)$
       \STATE Return $\tilde\Psi^h$
    \end{algorithmic}
\end{algorithm}

Here, the browser implements the sharing, while the the helpers perform the (partial) recovery of the summation of secret shares.  That is, the browser aims to decompose $y$ into two secret shares such that $y = \psi^0 + \psi^1$.  To do this, (e.g., the {\em map} operation) the browser implements Algorithm \ref{alg:secretshare}, splitting the shares such that each helper effectively receives a random value from a field $\mathcal R$.

For the corresponding {\em Reduce} operation, each helper receives a batch of records from the ad server and decrypts them.  It then executes Algorithm \ref{alg:summation_recovery}, where it sums the list of shares it receives.  The helpers additionally add noise, or apply $k$-anonymity, and return their partial sums.  For the final recovery, the advertising server sums the values returned from the helpers, e.g., $\tilde\Psi = \sum_h \tilde\Psi^h$.  Note the recovery algorithm is (intentionally) imperfect, due to the additive noise and the thresholds required.

\subsubsection{Query and Group By}
Often, an advertising service will wish to break apart the conversion by the corresponding campaign, or location, where the conversion happened.  In such a scenario, we allow the browser to pass additional side information to the helper services, to allow for filtering requests by the advertising service (e.g., what campaigns in Seattle have the most conversions).

The contextual information is exposed only to the helper services.  The helpers then repeatedly apply Algorithm \ref{alg:summation_recovery} on subsets of data matching the query, preserving $k$-anonymity and differential privacy with the conversion summations returned to the ad server.

We note that keys are revealed to the helper service, although the primary goal of preventing cross-site tracking is maintained due to the obfuscation of the corresponding label.  Thus, Masked LARk assumes an \textit{honest-but-curious} security model for the helpers.  In addition, Masked LARk charges the browsers to additionally insert fake records to cast doubt on the legitimacy of any individual record.

\section{Gradient Computation}\label{sec:gradient_computation}
Gradient computation is widely used in machine learning model training, and this is a new feature that is supported by the Masked LARk. We begin by introducing some extra notation.

\subsection{Notation}
Let $\mathbf X \subset \mathcal X$ indicate a set of user features drawn from the space of possible features, and $\mathbf Y \subset \mathcal Y$ indicate the set of user labels drawn from the space of possible labels.  Additionally $\mathcal X \subset \mathbb N^{d_x}$, $\mathcal Y \subset \mathbb N^{d_y}$ and both are bounded.  Let $\mathcal S = \mathcal X \times \mathcal Y$ indicate the sample space.  Individual samples for a user $u_i$ are denoted as $\mathbf s_i := (\mathbf x_i, \mathbf y_i) \in \mathcal S$. Let $\Theta$ indicate a parameter space (i.e., a model), with $\theta \in \Theta$ indicating a particularly chosen representation. 

Let $g$ indicate a bi-linear \emph{aggregation} function over a set of samples after applying a mapping $f: \mathcal X \times \mathcal Y \times \Theta \rightarrow \mathbb R^d$.  One such aggregation is summation:

$$
g\left(\left(f(\mathbf s_1, \theta), ..., f(\mathbf s_n, \theta)\right)\right) := \sum_{i \in n}f(\mathbf s_i, \theta).
$$

Our discussion and examples center around summation in particular, although the results generalize to other bi-linear functions.

\subsubsection{Real and Fake Labels}
A user (or a browser) is the final location where labels and features are combined to form a complete picture of the sample.  Once leaving the browser for storage at an ad server (and subsequent helpers), no single party should be able to reconstruct the complete \emph{real} sample $\mathbf s_i$.  Let $\mathbf \ynot_i \in \mathcal Y \backslash \{\mathbf y_i\}$ indicate a {\em fake} label that is additionally sent from the browser, forming a complete package of the form $(\mathbf x_i, \mathbf y_i, \mathbf \ynot_i)$.

There are additional constraints on the fake labels that must be considered.  

\begin{itemize}
\setlength\itemsep{-.5em}
    \item The labels are assigned by the browser based on whether a conversion was triggered: if no conversion is triggered, the default value will be assigned by the browser.  The default label, for labels $y \in \mathbf y$, is defined to be $0$.  Thus, one of the values $y_i, \ynot_i$ must be 0. 
    \item Conversely, there's no constraint that a real value {\em must} be sent; at times, the browser is required to send cases with strictly fake labels.
    \item We assume the ordering of real versus fake labels presented to the helpers is random, although the corresponding masks must match the label ordering
\end{itemize}

We generally assume non-zero fake labels are assigned by random draws of $\mathcal Y \backslash \{y\} $.  For real-valued labels this assumption can be too revealing.  To handle these types of labels, we employ probabilistic quantization to privatize the values.  First, we quantize a space $0\leq \Delta_{max}$ into $k$ evenly spaced buckets, with boundaries $0, q_1, ..., q_k, \Delta_{max}$. Let $y$ be the value to quantize, which must lie between some $q_i \leq y \leq q_{i+1}$.  We assign a quantized value $\tilde y$ like so:

\[
    \tilde y = 
\begin{cases}
    q_i     &   \text{with probability } (q_{i+1} - y) / (q_{i+1} - q_i)\\
    q_{i+1} & \text{otherwise}
\end{cases}
\]

Hence, $\tilde y = \mathbb{E}[y]$, however, the true value is privatized from downstream systems.  If $\tilde y$ is not an integer, we repeat the probabilistic thresholding to force $\tilde y$ to an integer.

\subsubsection{Masks}
Let $\alpha^h_i$ indicate a mask associated with a \emph{real} label $\mathbf y_i$, while $\alphanot^h_i$ indicates masks associated with a \emph{fake} label $\mathbf \ynot_i$. Similarly, $h$ indicates the helper utilizing the mask.  Without loss of generality, we assume only one real sample and one fake sample, as well as two helpers.  We define the following equalities:

$$
\alpha^0_i + \alpha^1_i = 1
$$
$$
\alphanot^0_i + \alphanot^1_i = 0
$$

Define a ring $\mathcal R := \mathbb Z / m \mathbb Z$, that is, a ring over a modulus of the discrete integers.  Simple examples are \code{int8} or \code{int16}.
We assume $\alpha^0_i, \alphanot^0_i \in \mathcal R$, and each are uniformly at random from the discrete uniform distribution over $\mathcal R$.  Combining the above, we have $\alpha^1_i = 1 - \alpha^0_i$ and $\alphanot^1_i = - \alphanot^0_i$.

\subsection{Masking Multiplication}
Consider a single data point received by a helper $h$, i.e., $(\mathbf x_i, \mathbf y_i, \mathbf \ynot_i, \alpha^h_i, \alphanot^h_i)$, along with a parameterization $\theta$ and a mapping $f$ to compute.  The corresponding other helper receives a matching set, with the same features and labels, but their own mask values.  The math below illustrates utilizing a distinction of real and fake values, however, we stress the helper is not aware of this distinction, and whether a value is real or fake is obfuscated from the helpers.

Each helper begins by computing the functions $f(\mathbf x_i, \mathbf y_i, \theta)$ and $f(\mathbf x_i, \mathbf \ynot_i, \theta)$, with the same results assumed computed by both helpers (for probabilistic functions, a shared seed must be used between the helpers).  Each helper next computes:
$$
\alpha^h_i \cdot f(\mathbf x_i, \mathbf y_i, \theta) + \alphanot^h_i \cdot f(\mathbf x_i, \mathbf \ynot_i, \theta)
$$
When combined, the resulting summation across the helpers becomes
\begin{equation}
\begin{split}
    \sum_h \left[\alpha^h_i \cdot \right. & \left.f(\mathbf x_i, \mathbf y_i, \theta) + \alphanot^h_i \cdot f(\mathbf x_i, \mathbf \ynot_i, \theta)\right] \\
    & = f(\mathbf x_i, \mathbf y_i, \theta) \sum_h \alpha^h_i  + f(\mathbf x_i, \mathbf \ynot_i, \theta) \sum_h \alphanot^h_i  \\
    & = f(\mathbf x_i, \mathbf y_i, \theta) \cdot 1  + f(\mathbf x_i, \mathbf \ynot_i, \theta) \cdot 0 \\
    & = f(\mathbf x_i, \mathbf y_i, \theta).
\end{split}
\label{eq:masks}
\end{equation}

We illustrate this in Figure \ref{fig:maskedFunctions}.  Equation \ref{eq:masks} then can be applied to multiple samples.  Our aim is to sum over the function values applied to the true samples.  Working backwards, we find

\begin{equation}
    \begin{split}
        \sum_i & f(\mathbf x_i, \mathbf y_i, \theta) = \\
        & = \sum_i \left[f(\mathbf x_i, \mathbf y_i, \theta) \cdot 1  + f(\mathbf x_i, \mathbf \ynot_i, \theta) \cdot 0\right] \\
        & = \sum_i \left[f(\mathbf x_i, \mathbf y_i, \theta) \sum_h \alpha^h_i  + f(\mathbf x_i, \mathbf \ynot_i, \theta) \sum_h \alphanot^h_i\right] \\
        & = \sum_i \left[\sum_h \alpha^h_i \cdot f(\mathbf x_i, \mathbf y_i, \theta) + \sum_h \alphanot^h_i \cdot f(\mathbf x_i, \mathbf \ynot_i, \theta)\right] \\
        & = \sum_h \left[\sum_i \alpha^h_i \cdot f(\mathbf x_i, \mathbf y_i, \theta) + \sum_i \alphanot^h_i \cdot f(\mathbf x_i, \mathbf \ynot_i, \theta)\right].
    \end{split}
\end{equation}

Thus, each helper can independently compute their sums and return the final summation.

The above has a more general formulation as well:

\begin{lemma}\label{lem:product}
Let $V$ be a $\mathcal R$-module equipped with a bilinear form $\langle \cdot, \cdot \rangle$.  There is a simple secret sharing protocol to compute $\langle v, w\rangle$ for $v, w \in V$ with three parties.
\end{lemma}

\begin{proof}
Suppose we have three parties, $H0$, $H1$ and $C$.  Let $w$ be held by $C$, and $H0, H1$ both hold $v$.

To compute $\langle v, w\rangle$, $C$ does an additive secret share protocol with $H_i$.  That is, $\alpha_j^0, \alpha_j^1 \in \mathcal R$, $\alpha_j^0 + \alpha_j^1 = w_j$.  $H_0$ computes $\Psi^0 = \sum_j \alpha_j^0$, $H_1$ computes $\Psi^1 = \sum_j \alpha_j^1$. Then $\langle v, w\rangle = \Psi^0 + \Psi^1$ as $\mathcal R$ has a bi-linear form.
\end{proof}

\begin{corollary}
Let $V$ be a $\mathcal R$-module equipped with a bilinear form $\langle \cdot, \cdot \rangle$.  Let $v \in V$, and $s \subseteq v$.
There is a simple secret sharing protocol to compute $\sum_{v_j\in s} v_j$ where $s$ is held by one party and $v$ by two other parties.
\end{corollary}

\begin{proof}
Apply Lemma \ref{lem:product} with $w := \mathcal X (s)$.
\end{proof}

\begin{figure}[t]
\begin{center}
\centerline{\includegraphics[width=\columnwidth]{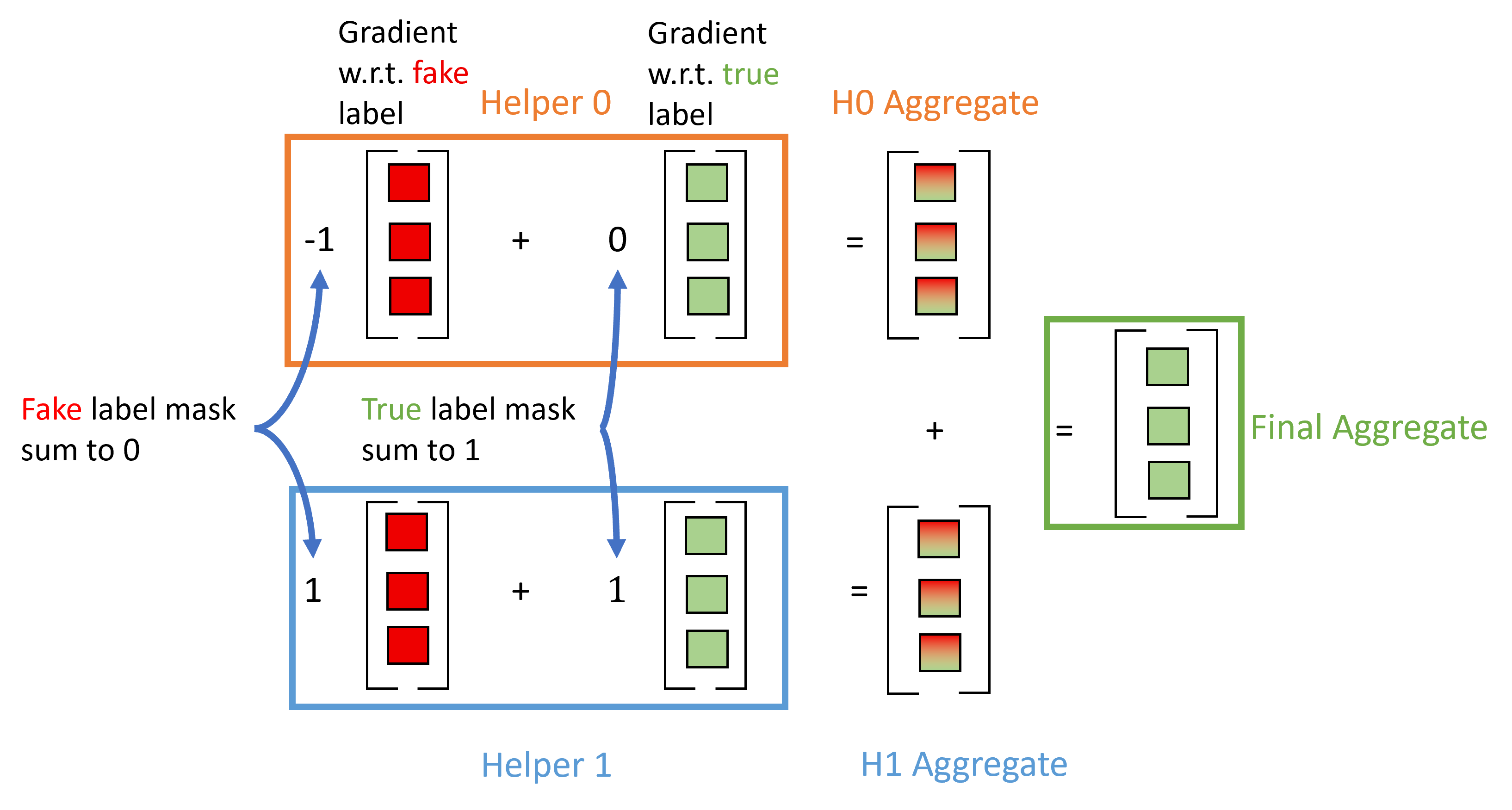}}
\caption{This diagram illustrates masking.  There are multiple labels (fake and real).  If summed vertically, the red simply cancels, and green remains.  However, as multiplication is distributed, we can instead have H0 and H1 sum its vectors and return to the advertising server.  When results are combined, only the true vector remains.}
\label{fig:maskedFunctions}
\end{center}
\end{figure}

\subsubsection{Masked Gradient Model Training}
With the general result above, we can apply a variety of functions $f$ to gather aggregate statistics over the sample space.  Possible $f$ mappings that are desired to compute include summation ($f(\mathbf x_i, \mathbf y_i, \theta) = \mathbf y_i$) and counting ($f(\mathbf x_i, \mathbf y_i, \theta) = 1$).  Each of these is handled with simpler mechanisms in subsection \ref{subsec:summations}.  Another interesting case, however, is model training.

For many machine learning paradigms, including supervised learning, we require a label $y_i$ and a feature vector $\mathbf x_i$.  The end goal is to minimize, in some way, the following \emph{loss function}:

$$
\argmin_\theta \sum_i L(F(\mathbf x_i), \mathbf y_i, \theta),
$$

where $L$ is some loss function and $F$ is some function from $\mathcal{X} \rightarrow \mathcal{Y}$, where $\mathbf{x}_i \in \mathcal{X}$ and $\mathbf{y}_i \in \mathcal{Y}$. Let's assume we have a differentiable model (e.g., a neural network) for $F$, parameterized by $\theta$.
A majority of neural network optimization follows some variant of Stochastic Gradient Descent \cite{bottou2010sgd}, using the following iterative approach to optimization:

$$
\theta_{j+1} = \theta_j - \eta \nabla F(x_i, y_i, \theta_j),
$$

where $\eta$ is a hyperparameter known as the learning rate.  In practice, we aggregate a number of data points at once using {\em minibatches} of data:

$$
\theta_{j+1} = \theta_j - \eta \sum_i \nabla F(x_i, y_i, \theta_j)
$$

Thus, we let $f := \nabla F$ in our above notation, meaning the {\em helpers} can compute masked aggregate gradients and send them back to the advertising server, and repeat.  There are a variety of per-coordinate learning rates, such as ADAM \cite{kingma2017adam}, AdaGrad \cite{duchi2011adagrad}, or AdaBound \cite{luo2019adabound}. These methods still rely on previously computed gradients, and utilizing MPC to compute the gradients works within these optimization methods as well.

\subsubsection{Global Gradient Noise}
The browser actions mirror the discussion above; each browser must create multiple labels and mask them before passing them to the helpers.  The privacy constraints, however, must be enforced by the helpers.  For this, we adapt an existing work \cite{song2013dpsgd} to compute gradients with (globally) differentially private updates. 

As discussed in Section \ref{sec:masked_lark}, we require both a $k$-anonymity constraint and differential privacy budget declared in a public parameter server. When a helper receives a batch of data and model to compute gradients over, it follows the algorithm described in Algorithm \ref{alg:helpergradient}.

\begin{algorithm}[t]
    \caption{Helper Computation}\label{alg:helpergradient}
    \begin{algorithmic}
        \STATE \textbf{Inputs}: gradient bound $\psi$, privacy parameter $\epsilon$, $k$ for $k$-anonymity
		\IF{minibatch size $<k$ }
		\STATE Return zero vector
		\ENDIF
		\FOR{each sample $(x_i, y_i, \alpha^h_i)$ in the minibatch }
		\STATE Compute $G(x_i, y_i) = \nabla L(F(\mathbf x_i), \mathbf y_i, \theta)$
		\STATE $\hat{G}(x_i, y_i) = \alpha^h_i \cdot G(x_i, y_i) / \max(\psi, |G(x_i, y_i)|)$
		\ENDFOR
		\STATE $\tilde{G}(x_i, y_i) = \hat{G}(x_i, y_i) + \eta$, where $\eta \sim Laplace(\frac{\psi}{\epsilon})$.
		\STATE Return $\sum_i \tilde{G}(x_i, y_i)$
    \end{algorithmic}
\end{algorithm}

We record the following theorem. 

\begin{theorem}[Theorem 3.6 \cite{dwork2014algorithmic}]
Let $\Delta$ denote the sensitivity of a function $f$. The Laplace mechanism with scale $\Delta/\epsilon$ preserves $(\epsilon, 0)$-differential privacy.
\end{theorem}

This ensures that the helper output of Algorithm \ref{alg:helpergradient} is differentially private with the privacy budget $\epsilon$.  We note other work \cite{abadi2016deep} relaxes the privacy guarantee to $(\epsilon, \delta)$-differential privacy by adding Gaussian noise instead of Laplacian noise, which could be explored in future work.

\subsubsection{Local Differential Privacy}
\label{subsubsec:localdp}
An alternative approach would be to enforce local differential privacy by adding noise to the data being sent to the helpers 
\cite{kasiviswanathan2011can}.  More precisely, given a feature vector $\mathbf x_i$, where each coordinate is bounded by some $\Delta$-sensitivity and there is a publicly available privacy parameter $\epsilon$, let each browser pass the following to the helpers\footnote{As a technical note, in practice we restrict feature values to a byte representation for space considerations, meaning the noise must be clipped at 0 and 255}:
\begin{equation}
    \tilde x_i = x_i + Laplace(0, \Delta / \epsilon)
\end{equation}

There are a number of trade-offs here in comparison to utilizing global differential privacy:
\begin{itemize}
\setlength\itemsep{0em}
    \item The feature vectors presented to the helpers are $\epsilon$-differentially private.
    \item So long as the collection mechanism is $\epsilon$-differentially private, then calculating gradients on the collected records is also $\epsilon$-differentially private \cite{dwork2006calibrating}.
    \item The helpers no longer need to clip gradients per sample (expensive computation) or add global noise to the sums.
    \item Local differential privacy is generally noisier than global differential privacy, with less signal available for learning. We explore the costs of this approach in Section \ref{sec:localprivacyexp}.
\end{itemize}

We pause to highlight that an ambitious adversary that is able to leverage both an advertising server and one of the helper services could work to recover a mapping between the locally differentially private vector and the true label by manipulating the input network and returned gradient.  We believe the effectiveness of this attack is limited, because of our k-anonymity constraint and the assumed limited cardinality of the label space.

\section{Experiments}\label{sec:experiments}

For demonstration purposes, we have implemented MaskedLARk as a training service for neural networks. In this section, we outline practical considerations that need to be made during practical implementation of MaskedLARk, as well as study the effect of a number of privacy-related parameters on model performance.

\begin{figure}[ht]
    \centering
    \subfigure[WBCD]{\label{fig:helperdp_bcd} \includegraphics[width=0.8\linewidth]{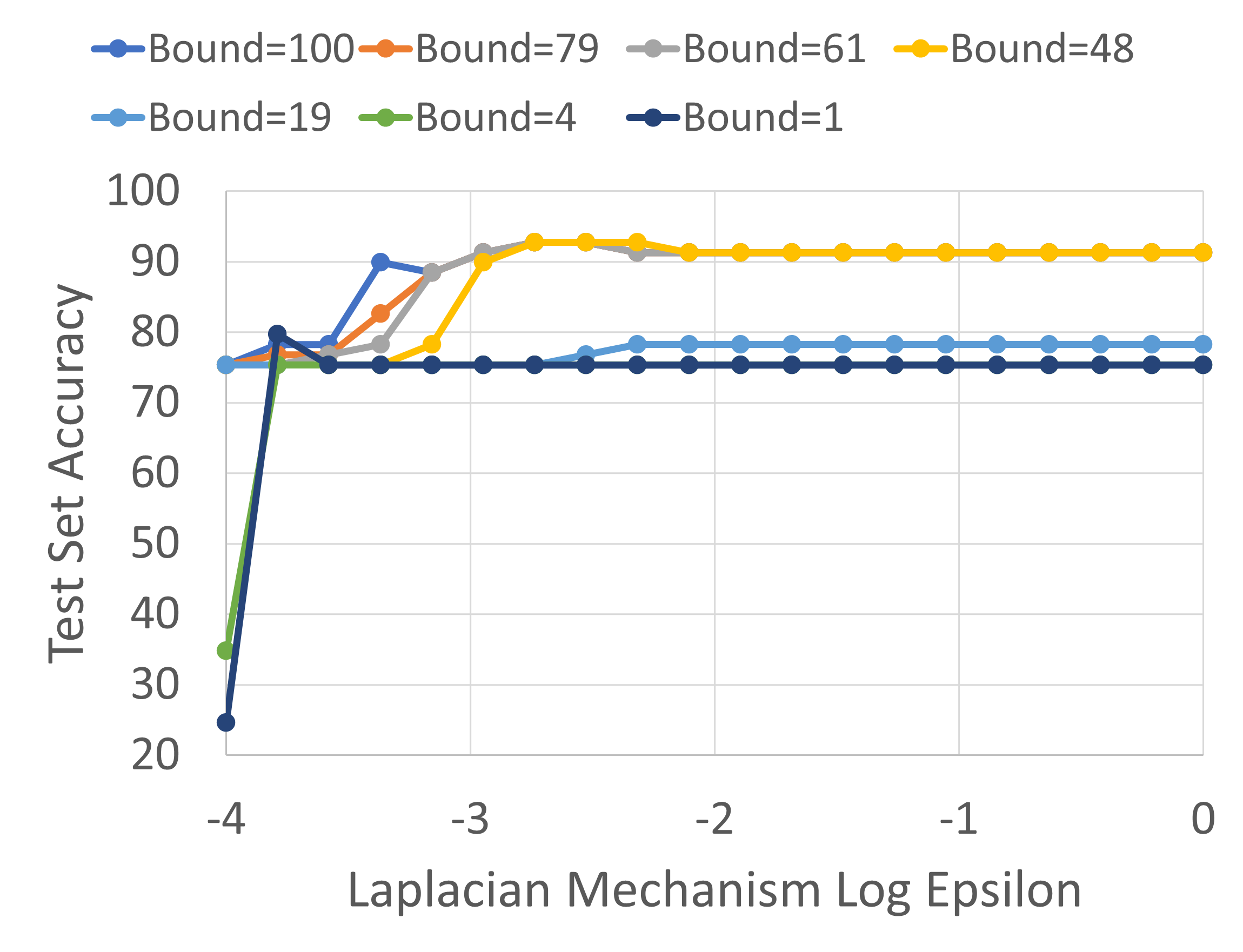}}\\
    \subfigure[MNIST]{\label{fig:helperdp_mnist} \includegraphics[width=0.8\linewidth]{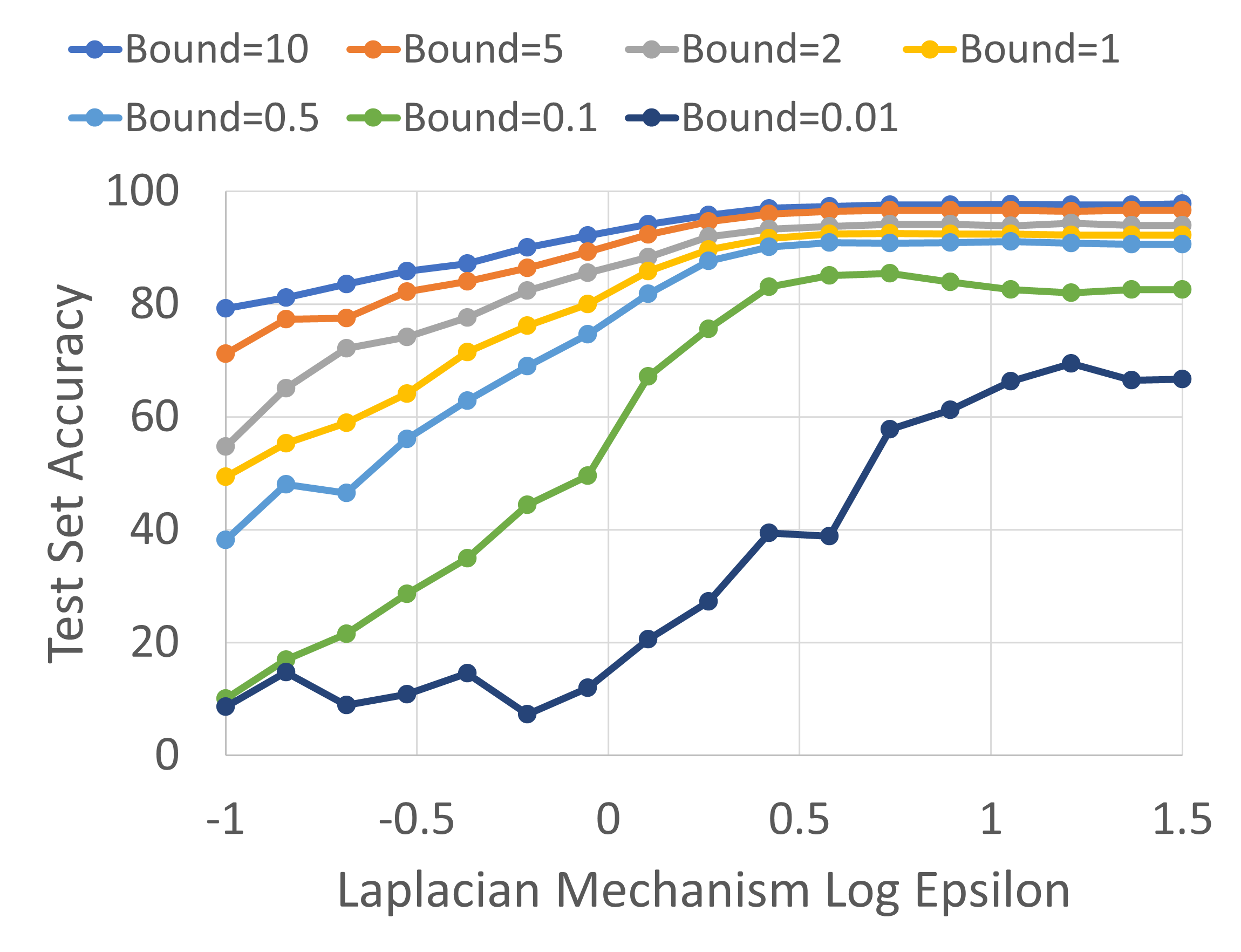}}~
    \caption{Both the gradient bound and Laplace noise power are important variables to set for privacy and performance purposes. Here we plot the effect of decreasing the magnitude of noise added to aggregated gradients on inference accuracy for a variety of gradient norm settings.}
    \label{fig:helperdp}
\end{figure}

\subsection{Setting}

We have implemented our demo MaskedLARk as a three-party system, in which a central node simulates the actions of both a browser and ad server, sending requests to two helper services which cannot communicate with each other. Here we focus on utilizing the helper services for gradient computation for the purposes of training a feedforward neural network. 

\subsubsection{Implementation Details}

We demonstrate MaskedLARk using two datasets: the MNIST digit classification dataset \cite{lecun1998gradient} and the Wisconsin Breast Cancer Dataset (WBCD) \cite{street1993nuclear}. Both of these datasets are standard classification benchmarks, while not requiring massive networks to achieve reasonable test accuracy. The MNIST dataset has 50000 training examples, input size of 784, and 10 classes, while the WBCD has 500 training examples, input size of 30, and a binary output. Notably, WBCD has a class imbalance, with roughly 75$\%$ of all datapoints being ``0'' instances. MNIST has perfectly balanced labels.

\begin{figure}[t]
    \centering
    \subfigure[WBCD]{\label{fig:requestinfo_wbcd} \includegraphics[width=0.8\columnwidth]{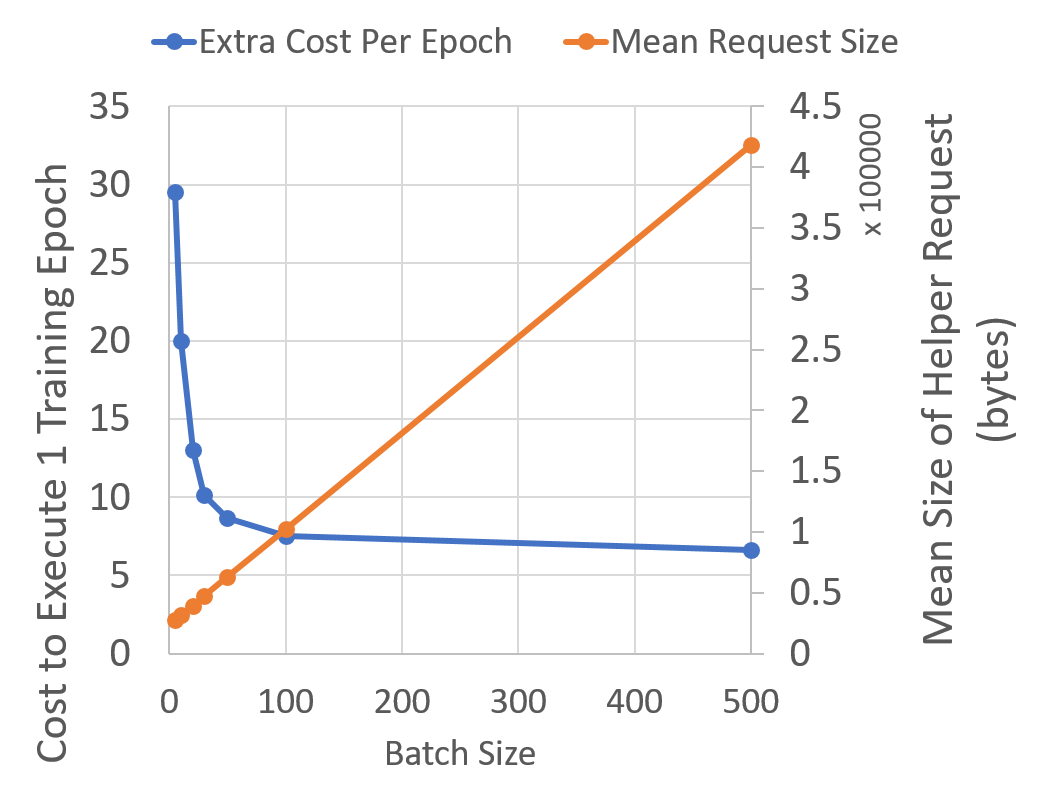}}\\
    \subfigure[MNIST]{\label{fig:requestinfo_mnist} \includegraphics[width=0.8\columnwidth]{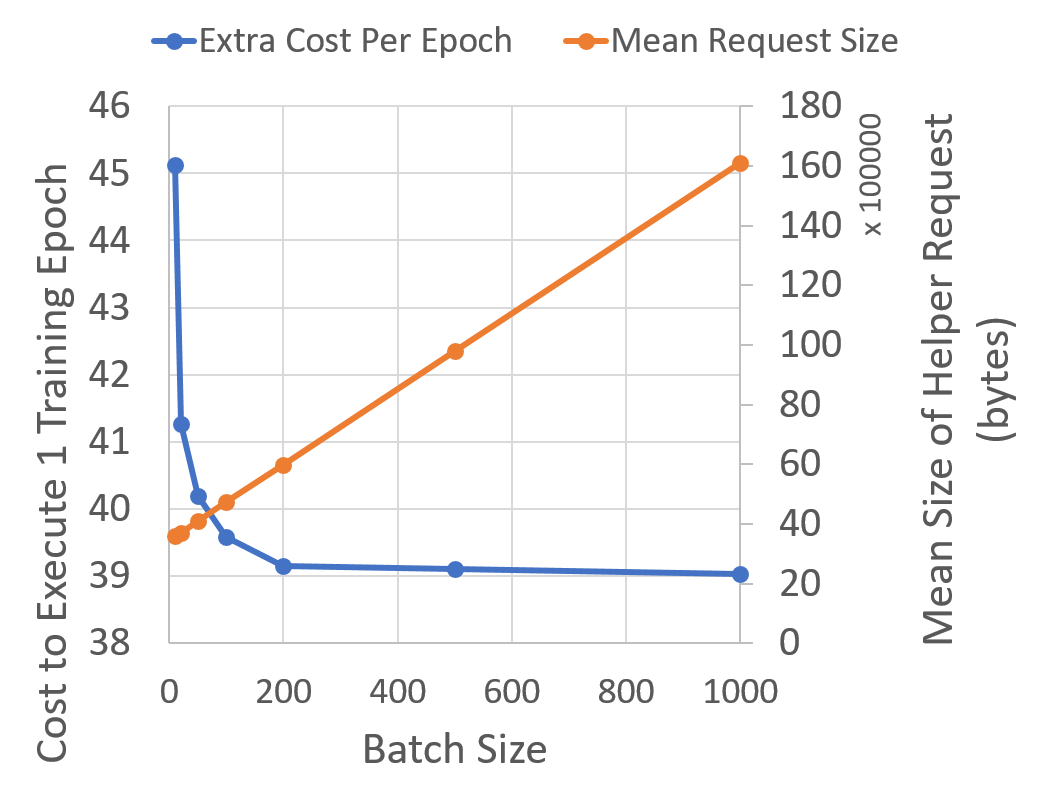}}~
    \caption{Making requests to helper services incurs a latency cost to training. Here we present the multiplicative latency increase from training our network via helper service versus loading all data locally. For small amount of data sent, most requests consist mainly of encoding the network, rather than the data. In general, we recommend using small networks and sending many data points at once, both for $k$-anonymity purposes as well as latency.}
    \label{fig:requestinfo}
\end{figure}

Our networks are simple feedforward networks with ReLU activations; the MNIST network has one hidden layer with 500 nodes, and the WBCD network has two hidden layers with 50 nodes. We note that the MNIST data and network demonstrated here are higher-dimensional than the intended use-case for MaskedLARk.

All gradient calculation is done in PyTorch on the helper and central processor side, and training is done via simple Stochastic Gradient Descent. To communicate between central processor and helper, all networks are encoded as ONNX models \cite{bai2019}. For latency and space reasons all network parameters are compressed to 8-bit values before sending (see Section \ref{sec:latencyexp} for reasoning).

\subsection{Effect of Privacy Mechanisms on Model Performance}\label{sec:helperdp}

As users prepare to deploy differential privacy mechanisms, a central concern is the effect of privacy on downstream model performance. In this Subsection, we explore the role of gradient clipping and the magnitude of Laplacian noise added to the gradients after aggregation. In Figure \ref{fig:helperdp} we plot the trade-off for a range of gradient bounds.

We observe nontrivial effects on model performance when gradient clipping is used, as well as detrimental effects for high levels of Laplacian noise. However, there are wide ranges of parameters at which competitive test-time accuracy can be achieved, so moderate privacy constraints need not require sacrificing useful models.

\begin{figure}[t]
    \centering
    \subfigure[WBCD]{\label{fig:localdp_wbcd} \includegraphics[width=0.8\linewidth]{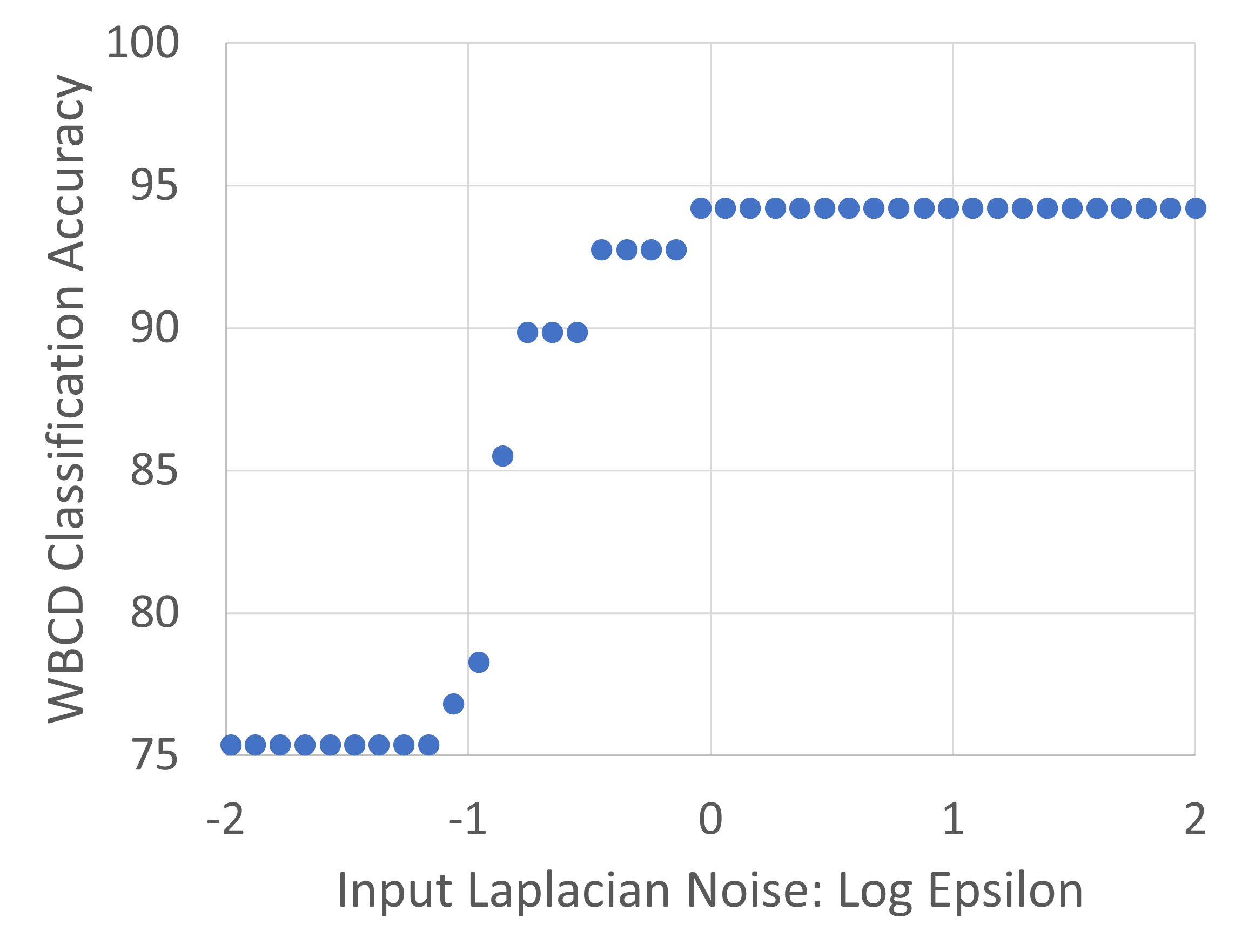}}\\
    \subfigure[MNIST]{\label{fig:localdp_mnist} \includegraphics[width=0.8\linewidth]{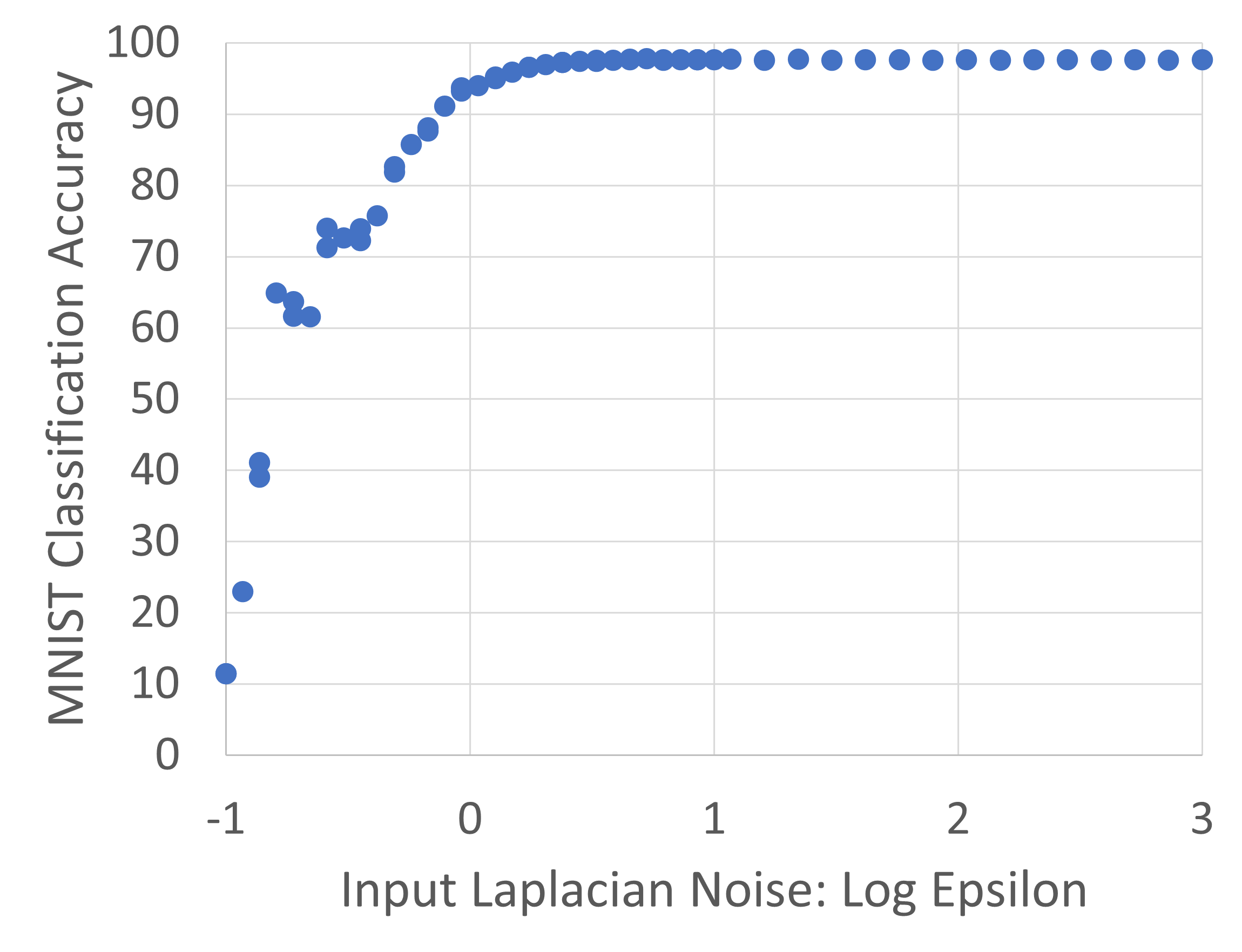}}~
    \caption{Simple neural networks are quite resilient to Laplacian noise injected into input features as part of a local differential privacy mechanism. Above we plot the network test accuracy against the log $\epsilon$ of injected Laplacian noise with variance $1/\epsilon$, so moving to the right in the above plots decreases noise power.}
    \label{fig:localdp}
\end{figure}

\subsection{Latency and Training Time}\label{sec:latencyexp}

Here we explore the latency effects of moving gradient computation to remote helpers. When computation must occur at the helper level, efficient communication is of a primary concern. In Figure \ref{fig:requestinfo}, we explore the effect of changing the number of data points sent to the helpers in a single request. 

We compare the amount of time required to train a network utilizing a single CPU with 64 GB of RAM, capable of storing the entire network and dataset in-memory. We duplicate the algorithm used by the helper, including gradient clipping and noise injection, so that the extra time comes only from packaging requests, parsing requests, and data transmission.

The most obvious takeaway is that sending as much data as possible at the same time is desirable from the perspective of training over the whole dataset. Small batch sizes will often be disallowed by $k$-anonymity rules, and as Figure \ref{fig:requestinfo} suggests, batch sizes that are too small also incur significant latency costs. This phenomenon is a consequence of sending data coupled with the networks that process that data; sending small number of data will result in the network needing to be compressed and parsed more frequently.

However, recent work in training deep networks suggests that batch sizes that are too large may result in poor generalization at inference-time \cite{keskar2017large}. In addition, practical helpers will have to place limits on requests, to prevent slow gradient calculation turnaround when multiple parties are making requests.

The size of the network being sent is central to getting timely results. The network we used for MNIST is intended to be a stress-test, and is far too large for practical use. Since we use the ONNX protobufs as the method to communicate the networks, we recommend \emph{quantization} of network parameters. Quantization of deep networks can range from simple limitations of the bit depth of parameters \cite{courbariaux2014training} to more complex schemes \cite{gong2014compressing, han2015deep}, any of which may be applied here, as long as the helper does not need to modify its protocol.

\subsection{Adding Noise to Data Directly: An Alternative Approach to Privacy} \label{sec:localprivacyexp}

While there are many advantages to the global differential privacy approach to Masked LARk, we observe that training of large networks with massive datasets may incur prohibitive delays due to point-wise application of gradient clipping during gradient calculation. Recently, some techniques have been proposed to improve the speed of this process \cite{lee2021scaling}, but this step remains a major bottleneck.

A reasonable alternative to clipping gradients and adding noise after summation is to instead add differential privacy noise \emph{before} gradient calculation. This scheme is referred to as local Differential Privacy \cite{kasiviswanathan2011can}, so that privacy can be handled browser-side for gradient calculation. So long as the collection mechanism is $\epsilon$-differentially private, then calculating gradients on the collected records is also $\epsilon$-differentially private \cite{dwork2006calibrating}.

We test this system, again leveraging the Laplace mechanism, but applying noise to the network inputs. Concretely, we send the helpers triplets $(x_i+\eta, y_i, \alpha_i)$, where $\eta \sim Laplace(\frac{1}{\epsilon})$. The helpers then return the sum $\sum_i \alpha_i \nabla L(f(x_i + \eta), y_i)$. We plot the performance of the trained networks as a function of $\epsilon$ in Figure \ref{fig:localdp}.

Figure \ref{fig:localdp} can be compared to Figure \ref{fig:helperdp}, since they both compare the effect of privacy settings on trained network performance. However, local differential privacy (in this section) often requires more noise to be injected than global differential privacy (in Section \ref{sec:helperdp}). With that in mind, we observe that the deep networks trained in our experiments are quite resilient to noise at the feature level, even with quite high variances compared to the scale of the features.

\section{Discussion}\label{sec:discussion}
In this paper, we have proposed Masked LARk, a Masked Learning, Aggregation and Reporting worKflow.  Masked LARk has several interesting properties:

\begin{itemize}
    \item It extends the existing proposal in \cite{google2020aggregatereporting} to a wide variety of tasks, outside of simple aggregation reporting, and most notably to enable model training.
    \item We discussed how differentially private \cite{dwork2014algorithmic} aggregated statistics can be computed and sent to the advertising server using browsers and helpers.  This enables user privacy protection, while allowing advertising services to show relevant ads to interested consumers.
    \item We outlined the usage of Masking, and showed how it generalizes to a wide variety of tasks.  Notably, this allows us to compute gradients for differentiable models, and train conversion models using secure MPC.
    \item In Section \ref{subsubsec:localdp}, we showed if we utilize local differential privacy on our feature vectors, the data sent to the helpers is differentially private, and by extension the learning algorithm is also differentially private.
    \item We perform several experiments on benchmark datasets, exploring the tradeoffs between differential privacy and accuracy. We plan to release a Python package with example usage of Masked LARk and a number of utilities for model training.
\end{itemize}

So far, the endpoint helper services are set up and can run aggregation and basic model training.  However, additional work needs to be done:

\begin{itemize}
    \item A public parameter service needs to be set up, where helpers can access publicly declared privacy settings for advertising servers
    \item Data sanitation is largely ignored at this time.  There does exist quite a bit of work in this space \cite{rathee2020cryptflow2}, meaning we need to add additional functionality to Masked LARk to enable it.
    \item Additional proposals for aggregate measurement exist \cite{msr2021bucketization}, and we need community feedback on which version of aggregation should be implemented.  To date, no other proposals exist for model training.
\end{itemize}

To this end, we will continue to work with the advertising community to agree upon standards which guarantee user privacy.
\bibliographystyle{ieeetr}
\bibliography{MaskedLARk}

\end{document}